\documentclass[sigconf,nonacm]{acmart}

\usepackage{array}
\usepackage{cleveref}

\AtEndPreamble{%
  \theoremstyle{acmplain}%
  \newtheorem{observation}[theorem]{Observation}}

\newcommand{\defeq}{\mathrel{\mathop:}=}
\newcommand{\signed}[1]{{\mathrm{sgn}}(#1)}
\newcommand{\signRank}[1]{{\mathrm{rank}}_{\pm}(#1)}
\newcommand{\relset}[2]{\text{rel}_{\text{set}}(#1, #2)}
\newcommand{\distset}[2]{\text{dist}_{\text{set}}(#1, #2)}

\begin{document}

\title{Beyond Dimensionality: Domain Shift and Catastrophic Forgetting in Single- vs.\ Multi-Vector Retrieval}

\author{Archish S}
\affiliation{%
  \institution{Microsoft Research India}
  \city{Bengaluru}
  \country{India}}
\author{Mihir Agarwal}
\affiliation{%
  \institution{Microsoft Research India}
  \city{Bengaluru}
  \country{India}}
\author{Ankit Garg}
\affiliation{%
  \institution{Microsoft Research India}
  \city{Bengaluru}
  \country{India}}
\author{Neeraj Kayal}
\affiliation{%
  \institution{Microsoft Research India}
  \city{Bengaluru}
  \country{India}}
\author{Kirankumar Shiragur}
\affiliation{%
  \institution{Microsoft Research India}
  \city{Bengaluru}
  \country{India}}

\begin{abstract}
Recent work by Weller et al. introduced the naturalistic benchmark LIMIT and showed that widely used single-vector embedding models suffer substantial degradation in retrieval quality, attributing this primarily to limited embedding dimensionality. In our work, we revisit this explanation and argue that \emph{dimensionality alone is unlikely to account for the observed failures}; in fact, results of Alon et al. imply that $2k+1$ dimensions suffice for top-$k$ retrieval over a fixed set of queries and documents. Motivated by this, we investigate alternative drivers of poor performance, including tokenization, linguistic structure, and domain shift. Our empirical analysis shows that tokenization artifacts and linguistic similarity yield only modest improvements, whereas domain shift and misalignment between embedding similarity and task-specific relevance are the dominant factors. While finetuning partially mitigates these issues and improves recall, it introduces \emph{catastrophic forgetting}: finetuning single-vector models on LIMIT-like datasets reduces their MS MARCO $\mathrm{Recall}@10$ by more than half, whereas multi-vector models forget much less. We observe the same pattern when finetuning on the real-world HotpotQA dataset, and between the single-vector and multi-vector heads of Jina-v4, which share one backbone.
\end{abstract}




\maketitle

\section{Introduction}\label{sec:intro}

\begin{figure*}[t]
\centering
\includegraphics[width=\textwidth]{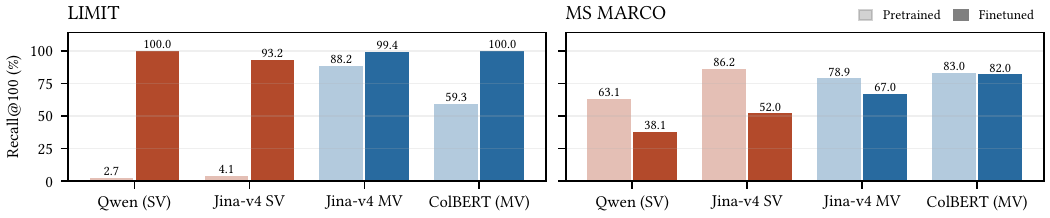}
\caption{$\mathrm{Recall}@100$ on LIMIT (left) and MS MARCO (right) before and after finetuning on LIMIT Train. LIMIT Train's 848 queries ask about attributes that appear in no LIMIT test query. Qwen (Qwen3-Embedding-0.6B) is a single-vector model and ColBERT (GTE-ModernColBERT-v1) a multi-vector one; the single-vector (SV) and multi-vector (MV) heads of Jina-v4 are finetuned jointly.}
\Description{Two bar charts of Recall@100 before and after finetuning on LIMIT Train for Qwen, the single-vector and multi-vector heads of Jina-v4, and ColBERT. On LIMIT, Qwen rises from 2.7 to 100.0 percent, the Jina-v4 single-vector head from 4.1 to 93.2 percent, its multi-vector head from 88.2 to 99.4 percent, and ColBERT from 59.3 to 100.0 percent. On MS MARCO, Qwen falls from 63.1 to 38.1 percent and ColBERT from 83.0 to 82.0 percent; the Jina-v4 single-vector head falls from 86.2 to 52.0 percent and its multi-vector head from 78.9 to 67.0 percent.}
\label{fig:teaser}
\end{figure*}

What are the most relevant documents for a given query? This question lies at the heart of information retrieval. Over the past decade, a family of methods known as dense retrieval has become prominent. Dense retrieval learns neural models that map both queries and documents to vector embeddings in a shared $d$-dimensional Euclidean space. Similarity, typically measured via cosine similarity, serves as a proxy for relevance and enables retrieval by nearest-neighbor search; when vectors are normalized, cosine similarity is equivalent to Euclidean distance up to a monotone transformation\footnote{Some communities prefer cosine similarity for simplicity, while others (e.g., vector indexing) favor Euclidean distance for its geometric intuition and connections to various algorithmic techniques.}. Dense retrieval has gained traction because embeddings can capture aspects of semantic meaning beyond lexical overlap. It has achieved strong performance in applications such as web search, especially when combined with complementary retrieval signals and efficient vector indexes, and vector embeddings are also widely used in image search and ad selection.

This success has fueled extensive work on embedding architectures, training objectives and procedures, and scalable indexing methods. Despite this progress, our conceptual understanding of the capabilities and limitations of single-vector embeddings remains incomplete. Recent work~\cite{weller2025} introduced LIMIT, a naturalistic but synthetically generated dataset, and found that single-vector embeddings produced by existing state-of-the-art neural models perform strikingly poorly on it. Given the widespread adoption and strong real-world results of single-vector embeddings, this failure mode has drawn significant attention. Why do vector embeddings break down so severely on LIMIT despite performing well in many practical settings? In this work, we investigate several candidate explanations for the difficulty of the LIMIT benchmark, including the hypothesis proposed by \citet{weller2025}, and contrast single-vector embeddings with multi-vector embeddings, which represent each query and document by one vector per token~\cite{khattab2020colbert,santhanam2022colbertv2}. We argue that dimensionality alone is unlikely to account for the poor performance, and find that domain shift and misalignment between embedding similarity and task-specific relevance are the dominant factors. We further study the role of finetuning in mitigating these challenges on LIMIT. Moving beyond LIMIT, we examine how finetuning can lead to \emph{catastrophic forgetting}~\cite{mccloskey1989catastrophic,kirkpatrick2017overcoming}, and find that single-vector models forget much more than multi-vector models, including when both share one backbone, as the heads of Jina-v4 do.
\paragraph{Relevance to the Web.} This paper addresses the Search, Recommendation, and Retrieval-Augmented AI track. Embedding-based retrievers are used for first-stage retrieval in Web search and in retrieval-augmented generation (RAG) systems that ground language models in Web content, and their models are regularly finetuned for new domains and tasks. Understanding when single-vector embeddings fail, whether finetuning can repair them, and what finetuning costs on previously learned retrieval tasks is therefore directly relevant to how Web-scale retrieval systems are built.
\paragraph{LIMIT dataset.} The LIMIT dataset introduced by \citet{weller2025} creates a simple natural language task with $q=1{,}000$ queries and $n=50{,}000$ documents. Documents correspond to profiles of people with a list of attributes they like, e.g., ``Jon Durben likes Quokkas and Apples''. In this example, the list of attributes liked by Jon has length $\ell = 2$, but in the actual dataset we have $\ell = 45$ for all persons. Queries are simple questions asking who likes a specific attribute, e.g., ``Who likes Quokkas?''. Document $D$ is relevant to a query $Q$ if the attribute in $Q$ is contained in $D$. The queries and documents in LIMIT are thus essentially sets over a universe of attributes $U$, where each query $Q$ has length $\ell_Q = 1$ and each document $D$ has length $\ell_D = 45$. In this case there is a natural ground-truth measure of relevance: for a query set $Q$ and a document set $D$, the underlying combinatorial relevance is the fraction of attributes in $Q$ that are present in $D$, i.e.,
\begin{equation}\label{eqn:setDistance}
   \text{rel}_{\text{set}}(Q, D) := \frac{|Q \cap D|}{|Q|}.
\end{equation}
The corresponding notion of distance is the fraction of elements of $Q$ that are absent in $D$, i.e., $|Q \setminus D|/|Q|$. To understand the underlying phenomenon, we also study datasets with different attribute universes $U$ and differing values of $\ell_Q$ and $\ell_D$. Accordingly, we say that a dataset is a \emph{LIMIT-like dataset} if all queries and documents correspond to sets (possibly of differing lengths and over a different set of attributes) and the relevance of a document $D$ to a query $Q$ is determined by $\text{rel}_{\text{set}}(Q, D)$ being above a certain threshold (possibly query-dependent).

Beyond introducing the dataset, \citet{weller2025} observed that embeddings produced by state-of-the-art neural models perform strikingly poorly on it. To explain this behavior, their theoretical analysis connects the dimensionality required to capture relevance information with an analytic quantity known as the sign rank of a matrix formed from the relevance data (see \Cref{defn:signRank}). In the worst case, relevance matrices, which may not arise in practical settings, can have high sign rank. Through this connection, high sign rank implies that embeddings must be high-dimensional for cosine similarity to capture the relevance matrix. Building on this insight, the authors suggest that the relatively low dimensionality of embeddings used in practice may explain the poor performance of popular neural models on LIMIT, hinting that the LIMIT dataset itself may have high sign rank.

\subsection{Our Contributions}
In our work, we develop a deeper understanding of the limitations of single-vector embeddings and study how various parameters influence their performance. We further contrast these findings with multi-vector embeddings, arguing that, in the settings we study, multi-vector representations outperform single-vector embeddings across multiple aspects, although single-vector embeddings remain cheaper to store and index. As a first step, we focus on understanding the factors driving the degradation in embedding model performance on the LIMIT dataset, beginning with the role of embedding dimensionality.

\paragraph{Role of dimensionality.} We first observe that when the relevance matrix is $k$-sparse, that is, for each query there are at most $k$ relevant documents, the sign rank of the matrix is upper bounded. In particular, results from \citet{amy16} imply that embeddings of dimension at most $2k+1$ suffice to represent such a relevance matrix; \citet[App.~D.3]{weller2025} and \citet{wang2026r2k} note the same bound. Since the LIMIT dataset has sparsity $k=2$, embeddings of dimension $5$ are sufficient for this dataset, far fewer than current models use. We further note that even when the total number of relevant documents is large, if the retrieval objective is to ensure that the top $k$ retrieved documents are relevant, a dimensionality of $2k+1$ again suffices. These upper bounds apply only to the problem of capturing the relevance matrix over a fixed, known set of queries and documents, and are most meaningful when $k$ is small. In practice, we need to learn models that generate embeddings for new queries and documents, which these bounds do not capture, and the construction can require arbitrarily high numerical precision~\cite[App.~D.3]{weller2025}. Limited dimensionality alone is therefore unlikely to explain the poor performance on LIMIT.

Additionally, we show that the classical Johnson--Lindenstrauss lemma~\cite{Johnson1984} yields an upper bound of $O(\min(\ell_Q^2 \cdot \ell_D^2 \cdot \log |U|,\, |U|))$ on the dimensionality required for a simple, canonical embedding model for LIMIT-like datasets, where $\ell_Q$ and $\ell_D$ denote the lengths of queries and documents, respectively. The bound holds for inner-product scoring, and for cosine similarity when all documents have the same length, as in LIMIT. Details appear in \Cref{sec:signRankUpperBounds}. In summary, the main takeaway is the following:

\begin{center}
\emph{Low-dimensional single-vector embeddings for LIMIT-like datasets exist whenever at least one of the following quantities is small: the sparsity of the relevance matrix, the number of documents to be retrieved, the length of queries and documents, or the size of the attribute universe.}
\end{center}

Following this result, we shift focus to complementary factors and study the role of tokenization, linguistic structure, domain shift, distance function misalignment and finetuning, in the context of LIMIT-like datasets, on the performance of both single- and multi-vector embeddings. Multi-vector embeddings, together with the associated Chamfer distance~\cite{khattab2020colbert, santhanam2022colbertv2}, provide richer representations, and in our experiments they outperform single-vector embeddings across these aspects.

\paragraph{Role of tokenization and linguistic aspects of attributes.} The LIMIT dataset was designed to be natural. While the chosen attributes are natural, this choice poses challenges for modern embedding models: many attributes are tokenized into multiple tokens rather than a single token, leading to token overlap across different attributes. Moreover, certain attribute pairs, such as \emph{iced tea} and \emph{black tea}, are linguistically similar, so a query for one attribute can cause the retriever to fetch documents corresponding to the other.

To isolate the impact of these factors, we construct Atomic LIMIT, which replaces the attributes with nouns from a fixed vocabulary intended to simplify tokenization and to reduce the similarity between attributes. This change yields only modest improvements for single-vector embeddings; for example, $\mathrm{Recall}@10$ increases from approximately $1\%$ to $3\%$. In contrast, multi-vector embeddings exhibit substantial gains: $\mathrm{Recall}@2$ improves from $27\%$ on the LIMIT dataset to $97\%$ on Atomic LIMIT, highlighting the superiority of the multi-vector approach. To further test this, we track how the angles between embeddings of attributes change when we train single-vector and multi-vector models on LIMIT-like datasets. The changes of angles in the single-vector setting are erratic, with no clear pattern, whereas for multi-vector models there is a clear pattern, with the angles increasing by $35^{\circ}$ to $42^{\circ}$. See \Cref{sec:tokenization} for more details. In summary, the main takeaway here is the following:

\begin{center}
\emph{Tokenization and linguistic similarity impact the performance of embedding models on LIMIT-like datasets; controlling for these factors leads to modest gains for single-vector embeddings, while multi-vector embeddings show much greater improvements.}
\end{center}

\paragraph{Role of domain shift, distance function misalignment and finetuning.}
The documents in the LIMIT dataset are a fairly long (of length $\ell = 45$) list of attributes. While the queries and documents in LIMIT are perfectly valid and even natural English sentences, it seems rather rare for documents in real-world settings to be a long list of unrelated items, and it is reasonable to expect that current embedding models have seen few such examples during training. Hence the LIMIT dataset represents a \emph{domain shift}, with inputs that have the underlying combinatorial structure of a set system that is rare in real datasets; the ability of models to handle such shifts is known as \emph{out-of-distribution generalization}. Current embedding models may struggle to handle this form of domain shift, especially when it is coupled with a combinatorial relevance measure that may be poorly aligned with relevance patterns arising in natural queries and documents, thereby indicating a potential distance-function misalignment. \citet{weller2025} tested this by finetuning a model on a LIMIT training set built from attributes that appear in no test query; its $\mathrm{Recall}@10$ rose from near zero to at most 2.8\%, and they concluded that the poor performance is not due to domain shift. It is nevertheless natural to ask whether finetuning enables existing models to handle this domain shift and makes the cosine similarity between embeddings better aligned with the combinatorial relevance given by \Cref{eqn:setDistance}. We finetune existing single-vector embedding models on two training sets built from LIMIT: \emph{LIMIT Train}, whose 848 queries each ask about an attribute that appears in no test query, and \emph{Split LIMIT}, which trains on 800 of the 1,000 original queries and tests on the other 200. We observe that finetuning significantly improves the retrieval quality (\Cref{fig:teaser} and \Cref{tab:finetuning,tab:jina-limit-finetuning}). For example, $\mathrm{Recall}@10$ improves from about $1\%$ to about $40\%$ for the single-vector embeddings. At the same time, for multi-vector embeddings $\mathrm{Recall}@10$ improves from about $39\%$ to near exact with finetuning. Finetuning the two heads of Jina-v4 jointly on LIMIT Train similarly lifts the $\mathrm{Recall}@10$ of its single-vector head from $1.5\%$ to $83\%$. We also perform finetuning and evaluation on several other LIMIT-like datasets, including \emph{Two LIMIT} and \emph{Three LIMIT}, whose queries ask who likes any of two or three attributes; see \Cref{sec:datasets,sec:finetuning} for further details. In summary, the main takeaway is the following:

\begin{center}
\emph{The poor performance of embedding models on LIMIT stems from a domain shift and distance function misalignment. Finetuning substantially improves retrieval, but even after finetuning, single-vector models lag behind multi-vector embeddings, which achieve near exact retrieval quality.}
\end{center}

In our second set of results, we study the phenomenon of catastrophic forgetting and evaluate how both single-vector and multi-vector models are affected by it.

\paragraph{Catastrophic forgetting.}
Although finetuning helps on LIMIT-like datasets, finetuning on such datasets leads to degradation in performance on previously learned retrieval benchmarks. In particular, finetuning single-vector embedding models on LIMIT-like datasets leads to substantial forgetting on MS MARCO (\Cref{fig:teaser} and \Cref{tab:limit-forgetting}), with the $\mathrm{Recall}@10$ of Qwen dropping by more than $50\%$ after finetuning on LIMIT Train. In contrast, multi-vector embedding models are significantly more robust, and ColBERT experiences only about a $5\%$ drop after the same finetuning. Similar trends are also observed when finetuning on the real-world HotpotQA dataset. Since Qwen and ColBERT differ in architecture, size and pretraining as well as in representation, we also finetune Jina-v4, whose single-vector and multi-vector heads share one backbone and are trained jointly; its single-vector head again forgets much more than its multi-vector head (\Cref{fig:teaser,fig:jina-forgetting}). These observations further highlight the relative brittleness of single-vector embedding models under domain adaptation. See \Cref{sec:forgetting} for details. In summary, the main takeaway is the following:

\begin{center}
\emph{Finetuning on LIMIT-like and real life datasets can cause catastrophic forgetting on previously learned retrieval benchmarks. This degradation is substantially more severe for single-vector embedding models, whereas multi-vector models remain much more stable under the same finetuning.}
\end{center}

\section{Related Work}\label{sec:related}

\paragraph{Single-vector dense retrieval.}
A long line of work studies the expressiveness and empirical performance of dense retrieval models based on bi-encoders that map queries and documents to single-vector embeddings scored by cosine similarity. Dense Passage Retrieval (DPR)~\cite{karpukhin2020dense} helped catalyze widespread adoption of this paradigm, and subsequent training innovations such as ANN-mined hard negatives~\cite{xiong2021ance} and cross-encoder distillation or joint retriever--reranker training~\cite{qu-etal-2021-rocketqa,ren-etal-2021-rocketqav2} further improved in-domain effectiveness. At the same time, heterogeneous benchmarks such as BEIR~\cite{thakur2021beir} highlighted brittleness under distribution shift, motivating efforts aimed at stronger zero-shot generalization such as unsupervised contrastive pretraining (e.g., Contriever)~\cite{izacard2022contriever, lei2023relevance} and large-scale dual encoders (e.g., GTR)~\cite{ni-etal-2022-large}.

\paragraph{Limitations of single-vector embeddings.}
Despite these advances, multiple analyses show that single-vector dense retrievers can fail in systematic ways: performance may degrade due to tokenization and lexical sensitivity artifacts~\cite{ram-etal-2023-token, sciavolino-etal-2021-simple}, domain mismatch~\cite{thakur2021beir}, fine-grained queries or queries with logical operations~\cite{xu-etal-2025-dense, malaviya2023quest}, multi-intent queries and diverse targets~\cite{chen2025beyond}, and geometric or representational pathologies (e.g., anisotropy and related degeneracies) that constrain what cosine similarity can express~\cite{ethayarajh-2019-contextual,li-etal-2020-sentence,gao2021simcse}. More fundamentally, \citet{weller2025} formalize an inherent ``single-vector bottleneck'' for embedding-based retrieval, in which relevance is set-structured or combinatorial in a way that cannot be captured by any low-dimensional embedding, even with powerful encoders and contrastive training. They also construct a dataset (LIMIT) on which all popular dense retrieval models perform poorly, whereas the lexical retriever BM25 reaches a $\mathrm{Recall}@10$ of 100\%. We ask whether neural embedding models can learn this task and do not study lexical or hybrid retrievers.

\paragraph{Multi-vector retrieval.}
An alternative paradigm, multi-vector or late-interaction retrieval, addresses these expressiveness limits by representing each text as a set of token-level vectors and scoring with a set-matching operator such as Chamfer similarity (MaxSim). ColBERT~\cite{khattab2020colbert} and ColBERTv2~\cite{santhanam2022colbertv2} instantiate this approach and substantially improve ranking quality over single-vector methods on many benchmarks. While late interaction increases indexing and scoring cost, systems work has narrowed the practicality gap via specialized execution and compression (e.g., PLAID)~\cite{santhanam2022plaid} and reductions of multi-vector search to single-vector ANN primitives (e.g., MUVERA)~\cite{jayaram2024muvera}. These developments complement our findings: multi-vector representations paired with Chamfer-style scoring capture the set-structured similarity underlying LIMIT-like tasks better than single-vector models.

\paragraph{Catastrophic forgetting.}
Neural networks trained sequentially on new tasks tend to lose previously acquired abilities, a phenomenon documented in early connectionist models~\cite{mccloskey1989catastrophic} and a central obstacle in continual learning~\cite{kirkpatrick2017overcoming}. For retrieval, it means that adapting a model to a new task or domain can degrade retrieval on data that it previously handled well; \citet{lovonmelgarejo2021studying} study this effect for neural ranking models trained sequentially on different domains. We instead compare single-vector and multi-vector representations under the same finetuning, both across models and between the two jointly trained heads of a single model, by tracking retrieval on MS MARCO~\cite{bajaj2016msmarco} after finetuning on LIMIT-like datasets and on HotpotQA~\cite{yang2018hotpotqa}.

\section{Upper Bounds on Dimensionality}\label{sec:signRankUpperBounds}

\paragraph{Notation.} Given a set of $q$ queries and $n$ documents, the ground-truth relevance matrix $R \in \{0, 1\}^{q \times n}$ captures the relevance information in the natural way: $R_{ij}$ is $1$ if document $j$ is relevant to query $i$ and $0$ otherwise.

\paragraph{Sign rank.} Let us recall the notion of \emph{sign rank}~\cite{amy16, weller2025}. It will enable us to reason about the dimensionality of embeddings needed to represent a relevance matrix $R$.

\begin{definition}[Signed Relevance Matrix and Sign Rank]\label{defn:signRank}
    For any real-valued matrix $M \in \mathbb{R}^{q \times n}$, the signed matrix of $M$, denoted $\signed{M}$, is defined as
    \begin{equation}
        \signed{M}_{ij} =
        \begin{cases}
            -1 & \text{if } M_{ij} \leq 0, \\
            +1 & \text{if } M_{ij} > 0.
        \end{cases}
    \end{equation}
    For a binary relevance matrix $R \in \{0, 1\}^{q \times n}$ we have
    \[ \signed{R} \defeq (2 \cdot R - \mathbf{1}_{q \times n}), \]
    where $\mathbf{1}_{q \times n}$ denotes the $q \times n$ matrix all of whose entries are $1$. The sign rank of $R$, denoted $\signRank{\signed{R}}$, is defined as
    \begin{align*}
        \signRank{\signed{R}}
        &= \min_{M : \mathrm{shape}(M) = \mathrm{shape}(R)} \\
        &\quad \big\{ \mathrm{rank}(M) : \mathrm{sign}(M) = \signed{R} \big\}.
    \end{align*}
\end{definition}

Sign rank is a well-studied notion with many connections to machine learning. \citet{weller2025} apply it to retrieval and observe that $\signRank{\signed{R}}$ almost exactly captures the dimension required to embed the underlying queries and documents. Specifically:

\begin{proposition}[Proposition 2 of \cite{weller2025}]\label{prop:wellerSignRank}
    Given a binary relevance matrix $R$, the dimension of embeddings required to capture the relationships in $R$ is at least $\signRank{\signed{R}} - 1$ and at most $\signRank{\signed{R}}$.
\end{proposition}

In the LIMIT dataset each query has $k=2$ relevant documents. This means that the corresponding binary relevance matrix $R$ is \emph{$2$-row-sparse}, i.e., each row of $R$ has at most two nonzero entries. For such relevance matrices, a result of \citet{amy16} immediately bounds $\signRank{\signed{R}}$; \citet[App.~D.3]{weller2025}, using a construction of \citet{alon1985geometrical}, and \citet{wang2026r2k} give the same bound for retrieval. Specifically, we have:

\begin{observation}\label{obs:dimBoundkLIMIT}
    For the LIMIT dataset, there exist embeddings of dimension $d = 2 \cdot k + 1 = 5$ that can capture the query-document relevance information.
\end{observation}

\noindent For completeness, we include a proof below. For a retrieval task, if we only want the top-$k$ retrieved results to be relevant and do not care about the ranking of the remaining documents, then the corresponding relevance matrix can effectively be made sparse, so that the above observation again applies. Hence we also have:

\begin{observation}\label{obs:dimBoundk}
    For any dataset with $q$ queries and $n$ documents and any $k \geq 1$, there exist embeddings of dimension $d = 2k + 1$ such that, for every query with $r$ relevant documents, the $\min(k, r)$ documents with the highest cosine similarity are all relevant.
\end{observation}

\noindent Popular embedding models produce embeddings with hundreds or thousands of dimensions, far more than five, so a lack of capacity in the embedding space alone is unlikely to explain their poor performance on LIMIT. Both observations are existence results for a fixed, known relevance matrix, however. They do not show that an encoder trained on data can find such embeddings or that the embeddings generalize to new queries and documents, and the construction can require arbitrarily high numerical precision~\cite[App.~D.3]{weller2025}. Whether a trained encoder can find a five-dimensional solution, even for the fixed LIMIT matrix, remains open.

\paragraph{Proof of Observations~\ref{obs:dimBoundkLIMIT} and~\ref{obs:dimBoundk}.}
Let $\mathcal{C}=(\mathcal{Q},\mathcal{D})$ be a corpus with $|\mathcal{Q}|=q$ queries and $|\mathcal{D}|=n$ documents.
Let $R_k \in \{0,1\}^{q\times n}$ be the binary relevance matrix, where $(R_k)_{i,j}=1$ iff document
$D_j$ is relevant to query $Q_i$. Assume that for every query $Q_i$, there are at most $k$ relevant
documents, i.e. each row of $R_k$ has at most $k$ ones.
Define the associated sign matrix
\[
S := 2R_k - \mathbf{1}_{q\times n} \in \{\pm 1\}^{q\times n},
\]
so that $S_{i,j}=+1$ iff $(R_k)_{i,j}=1$ and $S_{i,j}=-1$ otherwise.

\begin{proposition}\label{prop:2kplus1}
There exists a real matrix $M\in\mathbb{R}^{q\times n}$ with $\operatorname{rank}(M)\le 2k+1$
such that $\operatorname{sgn}(M)=S$ entrywise. In particular,
$\operatorname{rank}_{\pm}(S)\le 2k+1$.
Consequently (by \Cref{prop:wellerSignRank}), there exist embeddings $u_i,v_j\in\mathbb{R}^{2k+1}$ such that
\[
\operatorname{sgn}(\langle u_i, v_j\rangle)=S_{i,j}
\]
for all $i\in\{1,\dots,q\}$, $j\in\{1,\dots,n\}$.
Hence, if a retriever ranks documents by $\langle u_i, v_j\rangle$, then for each query $Q_i$ all
relevant documents receive positive score and all irrelevant documents receive negative score, so
top-$k$ retrieval achieves perfect recall whenever each query has $\le k$ relevant documents.
\end{proposition}

\begin{proof}
We follow the construction implicit in \cite{amy16}. Fix a row index $i\in\{1,\dots,q\}$ and consider
the sign pattern $(S_{i,1},\dots,S_{i,n})\in\{\pm 1\}^n$.
Since row $i$ of $R_k$ has at most $k$ ones, row $i$ of $S$ has at most $k$ entries equal to $+1$.
Therefore, row $i$ contains at most $k$ contiguous blocks of $+1$'s, and thus the number of sign
changes in the sequence $(S_{i,1},\dots,S_{i,n})$ is at most $2k$.

Let $t_{i,1} < t_{i,2} < \cdots < t_{i,s_i}$ be the indices in $\{1,\dots,n-1\}$ where a sign change
occurs, i.e. $S_{i,t_{i,\ell}}\neq S_{i,t_{i,\ell}+1}$, and note that $s_i\le 2k$.
Define half-integers
\[
r_{i,\ell} := t_{i,\ell}+\tfrac12 \qquad (\ell=1,\dots,s_i),
\]
and define the polynomial
\[
G_i(x) := S_{i,1}\cdot \prod_{\ell=1}^{s_i} (r_{i,\ell}-x).
\]
Then $\deg(G_i)=s_i\le 2k$, and $G_i(j)\neq 0$ for every integer $j\in\{1,\dots,n\}$ because all roots
are half-integers. Moreover, for integer $j$, the sign of the factor $(r_{i,\ell}-j)$ flips exactly when
$j$ passes from $t_{i,\ell}$ to $t_{i,\ell}+1$, so $\operatorname{sgn}(G_i(j))$ flips exactly at the
sign-change positions of the row. Since $G_i(1)$ has sign $S_{i,1}$ (all factors are positive at $x=1$),
we obtain
\[
\operatorname{sgn}(G_i(j)) = S_{i,j}\qquad\text{for all } j\in\{1,\dots,n\}.
\]

Now define a real matrix $M\in\mathbb{R}^{q\times n}$ by $M_{i,j}:=G_i(j)$. Then
$\operatorname{sgn}(M)=S$ entrywise. Also, since each $G_i$ has degree at most $2k$, we can write
\[
G_i(x)=\sum_{m=0}^{2k} a_{i,m}x^m,
\]
padding with zeros if $\deg(G_i)<2k$. Let $A\in\mathbb{R}^{q\times(2k+1)}$ be defined by
$A_{i,m+1}=a_{i,m}$ and let $V\in\mathbb{R}^{n\times(2k+1)}$ be the Vandermonde-type matrix
$V_{j,m+1}=j^m$. Then $M=AV^\top$, and hence $\operatorname{rank}(M)\le 2k+1$.
Therefore $\operatorname{rank}_{\pm}(S)\le 2k+1$ by definition of sign rank.

Finally, \Cref{prop:wellerSignRank} converts $\operatorname{rank}_{\pm}(S)\le 2k+1$ into the existence of
embeddings $u_i,v_j\in\mathbb{R}^{2k+1}$ with $\operatorname{sgn}(\langle u_i,v_j\rangle)=S_{i,j}$.
This implies that every relevant pair has positive inner product and every irrelevant pair has negative
inner product; since each query has at most $k$ relevant documents, top-$k$ retrieval returns all
relevant documents (perfect recall).
\end{proof}

\paragraph{Low-dimensional embeddings without a sparsity assumption.} The above discussion concerns the existence of low-dimensional embeddings \emph{when the relevance matrix is sparse}, or when we only care about the correctness of the top-$k$ retrieved results. Even without the sparsity assumption, if we work with LIMIT-like datasets in which queries and documents are sets of attributes, we can still construct low-dimensional embeddings.

To do this, we first construct a set of $|U|$ unit vectors in $d$-dimensional Euclidean space $\mathbb{R}^d$ that are pairwise \emph{almost} orthogonal. The Johnson--Lindenstrauss lemma~\cite{Johnson1984} implies that a random set of vectors has this property with high probability, but deterministic constructions based on error-correcting codes or set systems with small intersections (cf.~\cite{DeVore2007}) can also be used. We assign one vector from this set to each token in $U$. A query is embedded as the normalized sum of the vectors assigned to its tokens, and a document as the unnormalized sum. When $d \gg \min( \ell_Q^2 \cdot \ell_D^2 \cdot \log |U| ,\,|U|)$, the inner products of these embeddings induce, for every query, the same ordering on documents as the set-theoretic relevance function of \Cref{eqn:setDistance}, resulting in low-dimensional embeddings for inner-product retrieval. If documents are normalized as well, so that the score is cosine similarity, the same holds when all documents have the same length, as in LIMIT, whose profiles all list 45 attributes. Further details are in \Cref{sec:proof-prop5}.

\section{Experiments}\label{sec:empirical}

We evaluate and finetune single-vector and multi-vector models on LIMIT-like datasets (\Cref{sec:tokenization,sec:finetuning}) and then measure how much retrieval ability on MS MARCO these models lose when finetuned on LIMIT-like datasets and on the real-world HotpotQA dataset (\Cref{sec:forgetting}).

\subsection{Setup}\label{sec:setup}

\paragraph{Models.} We use two pretrained models through the Sentence Transformers library: the single-vector model Qwen3-Embedding-0.6B~\cite{zhang2025qwen3} (Qwen), which produces 1024-dimensional embeddings, and the multi-vector model GTE-Modern\-ColBERT-v1~\cite{chaffin2025gtemoderncolbert} (ColBERT), whose 128-dimensional token embeddings are scored by Chamfer similarity. Since these two models differ in architecture, parameter count and pretraining as well as in representation, we also evaluate jina-embeddings-v4 (4B parameters)~\cite{gunther2025jina4}, which produces a 2048-dimensional single-vector (SV) representation and 128-dimensional token-level multi-vector (MV) representations from a shared backbone.

\paragraph{Finetuning.} For Qwen and ColBERT, we use the Sentence Transformers package and follow the prescribed training routines there; these models are trained on a 4$\times$A100 system. For finetuning experiments on LIMIT, we use the AdamW optimizer with an initial learning rate of $10^{-7}$ together with a cosine learning rate scheduler for both models. We train all variants until retrieval performance (primarily recall) saturates, or terminate training after a maximum budget of 48 hours. We finetune Qwen using the standard Qwen fine-tuning pipeline. For each query, we construct a small training batch by pairing it with two positive documents and thirty randomly sampled negatives, and train for one epoch.

For Jina-v4 on LIMIT Train, we train only its retrieval adapter on a 4$\times$A100 system and optimize both heads jointly. Each training group pairs a query with two of its relevant profiles and 30 random negatives; each head receives a softmax cross-entropy loss over the group's 32 candidates (temperature 0.05), and a term $0.01\,D_{\mathrm{KL}}(p_{\mathrm{SV}}\Vert p_{\mathrm{MV}})$ couples the two heads. We use AdamW with a peak learning rate of $10^{-6}$ and a global batch of 32 groups, and report the checkpoint after 40,000 updates. For HotpotQA, each query has two supporting positives, 20 mined hard negatives and 10 random negatives.

\paragraph{Metric.} We report recall macro-averaged over queries,
\[
\mathrm{Recall}@K = \frac{|\mathrm{Top}_K(q)\cap R_q|}{\min(K,|R_q|)},
\]
where $\mathrm{Top}_K(q)$ is the set of the $K$ highest-scoring documents for query $q$ and $R_q$ is its set of relevant documents. This equals ordinary recall on Original LIMIT for $K\geq2$; on Two and Three LIMIT, whose queries have many relevant documents, it need not increase with $K$.

\subsection{Datasets}\label{sec:datasets}

\Cref{tab:datasets} summarizes the LIMIT variants we use, and \Cref{sec:allExpts} defines them in full.

\paragraph{Evaluation sets.} \emph{Original LIMIT} is the 1,000-query evaluation set of \citet{weller2025}. The two relevant profiles of each of its queries are drawn from 46 core profiles; the remaining 49,954 profiles contain only filler attributes from a separate vocabulary of 848 attributes. \emph{Atomic LIMIT} maps every attribute, in queries and profiles alike, to a noun from a fixed replacement vocabulary intended to simplify tokenization and to reduce linguistic similarity; it preserves the relevance structure and replaces person names with ``This person''. No model we train has seen its renamed attributes. \emph{Extended LIMIT} appends 50,000 profiles that contain only filler attributes, doubling the corpus.

\paragraph{Training sets.} \emph{LIMIT Train} asks about the 848 filler attributes, which are disjoint from the Original LIMIT query attributes, so Original LIMIT is held out from it. \emph{Split LIMIT} partitions the 1,000 original queries into 800 training and 200 held-out test queries. \emph{Permuted LIMIT} keeps the original query texts but permutes the query attributes across profiles, changing which profiles are relevant, and includes three orderings of each profile. \emph{Fresh LIMIT} keeps the original query texts with an independently sampled corpus. Models trained on Permuted or Fresh LIMIT have therefore seen the Original LIMIT query texts, paired with a different corpus and relevance assignment.

\paragraph{Multi-attribute queries.} \emph{Two} and \emph{Three LIMIT} ask which profiles like any of two or three attributes drawn from the full vocabulary of 1,848 attributes, which includes the Original LIMIT query attributes; relevance is the union of the attributes' relevance sets. Models trained on them have seen queries that contain the Original LIMIT attributes.

To measure forgetting, we evaluate models finetuned on LIMIT Train and on the HotpotQA dataset~\cite{yang2018hotpotqa} on MS MARCO~\cite{bajaj2016msmarco}.

\begin{table}[t]
\caption{LIMIT variants used in our experiments. Positives is the number of relevant documents per query; Split LIMIT has 800 training and 200 test queries.}
\label{tab:datasets}
\small
\begin{tabular}{@{}lrrr@{}}
\toprule
Dataset & Queries & Documents & Positives \\
\midrule
Original LIMIT & 1,000 & 50,000 & 2 \\
Atomic LIMIT & 1,000 & 50,000 & 2 \\
LIMIT Train & 848 & 50,000 & 2,484--2,792 \\
Split LIMIT & 800 / 200 & 50,000 & 2 \\
Permuted LIMIT & 1,000 & 150,000 & 6 \\
Fresh LIMIT & 1,000 & 50,000 & 2 \\
Extended LIMIT & 1,000 & 100,000 & 2 \\
Two LIMIT & 2,725 & 50,000 & 3--5,393 \\
Three LIMIT & 1,805 & 50,000 & 4--7,745 \\
\bottomrule
\end{tabular}
\end{table}

\subsection{Tokenization and Linguistic Similarity}\label{sec:tokenization}

In LIMIT, many attributes are split across multiple tokens, and different attributes can share lexical components while denoting distinct entities, for example, \emph{iced tea} and \emph{black tea}, or \emph{apple pie} and \emph{apple slices}. To test how much tokenization and linguistic similarity explain the poor performance of models on LIMIT, we evaluate the pretrained models on Atomic LIMIT, whose replacement nouns are intended to tokenize simply and to be linguistically dissimilar. \Cref{tab:zeroshot} shows that the single-vector model performs better on Atomic LIMIT than on Original LIMIT, but its recall remains low, whereas the multi-vector model is almost perfect on Atomic LIMIT. For Jina-v4, moving from Original to Atomic LIMIT raises the $\mathrm{Recall}@10$ of the SV head only from 1.5\% to 4.0\%, and that of the MV head from 77.2\% to 85.4\%. Controlling tokenization and linguistic similarity therefore improves single-vector models only modestly, while multi-vector models improve substantially.

\begin{table}[t]
\caption{Recall (\%) of pretrained models on Original and Atomic LIMIT. Qwen is single-vector (1024 dimensions); ColBERT is multi-vector (128 dimensions per token). The SV and MV heads of Jina-v4 share one backbone.}
\label{tab:zeroshot}
\small
\setlength{\tabcolsep}{3.5pt}
\begin{tabular}{@{}lrrrrrr@{}}
\toprule
& \multicolumn{3}{c}{Original} & \multicolumn{3}{c}{Atomic} \\
\cmidrule(lr){2-4}\cmidrule(l){5-7}
Model & R@2 & R@10 & R@100 & R@2 & R@10 & R@100 \\
\midrule
Qwen & 0.8 & 1.2 & 2.7 & 2.2 & 3.1 & 7.1 \\
ColBERT & 27.4 & 39.5 & 59.3 & 96.8 & 98.2 & 99.2 \\
\midrule
Jina-v4 SV & 0.8 & 1.5 & 4.1 & 2.7 & 4.0 & 8.9 \\
Jina-v4 MV & 63.6 & 77.2 & 88.2 & 76.2 & 85.4 & 91.9 \\
\bottomrule
\end{tabular}
\end{table}

To test this further, we track the angles between the embeddings of various attributes as we finetune models on LIMIT-like datasets. As \Cref{tab:attr-angles} shows, finetuning makes the multi-vector embeddings of these attributes consistently more orthogonal, increasing every angle by $35^{\circ}$ to $42^{\circ}$, whereas the angles of the single-vector model change erratically.

\begin{table}[t]
\caption{Angles (in degrees) between the embeddings of selected attribute pairs before and after finetuning. SV is Qwen and MV is ColBERT; the second header row gives the training dataset.}
\label{tab:attr-angles}
\small
\begin{tabular}{@{}>{\raggedright\arraybackslash}p{1.75cm}rrrrrr@{}}
\toprule
& \multicolumn{2}{c}{SV} & \multicolumn{2}{c}{SV} & \multicolumn{2}{c}{MV} \\
& \multicolumn{2}{c}{LIMIT Train} & \multicolumn{2}{c}{Permuted} & \multicolumn{2}{c}{LIMIT Train} \\
\cmidrule(lr){2-3}\cmidrule(lr){4-5}\cmidrule(l){6-7}
Attribute pair & Before & After & Before & After & Before & After \\
\midrule
\emph{apple pie},\newline \emph{apple slices} & 39.6 & 43.3 & 39.6 & 47.8 & 14.7 & 55.4 \\
\emph{black tea},\newline \emph{iced tea} & 33.1 & 47.4 & 33.1 & 36.4 & 16.4 & 58.6 \\
\emph{feta cheese},\newline \emph{wrestling} & 61.6 & 44.8 & 61.6 & 52.9 & 20.2 & 61.6 \\
\emph{astronomy},\newline \emph{the space age} & 47.0 & 54.6 & 47.0 & 45.3 & 15.0 & 50.5 \\
\bottomrule
\end{tabular}
\end{table}

\subsection{Domain Shift, Distance-Function Misalignment and Finetuning}\label{sec:finetuning}

Finetuning substantially improves retrieval on LIMIT-like datasets (\Cref{tab:finetuning}), but for the single-vector model the gains depend strongly on the training condition. Training Qwen on LIMIT Train, Split LIMIT or Permuted LIMIT lifts $\mathrm{Recall}@50$ from 2\% to essentially 100\%, yet $\mathrm{Recall}@2$ stays below 13\%: the model learns to place the two relevant profiles among its top fifty but not to rank them first. Fresh LIMIT and Two LIMIT trade the other way, raising $\mathrm{Recall}@2$ to 19\% and 28\% while $\mathrm{Recall}@100$ saturates near 86\%. Two LIMIT is the most transferable training set, although its queries already contain the Original LIMIT attributes: its gains carry over to a corpus of twice the size (Extended LIMIT, $\mathrm{Recall}@10$ of 53\% versus 57\%), to three-attribute queries (Three LIMIT, $\mathrm{Recall}@10$ from 27\% to 83\%), and it is the only condition that transfers to Atomic LIMIT at all, whereas Qwen trained on LIMIT Train has zero recall there. The multi-vector model reaches near-perfect recall in every condition we report, including LIMIT Train to Atomic LIMIT.

\begin{table}[t]
\caption{Recall (\%) after finetuning on LIMIT-like datasets. Pretrained is the model before finetuning; a blank evaluation cell repeats the one above. The Split row is trained on the 800 Split LIMIT training queries and evaluated on its 200 test queries.}
\label{tab:finetuning}
\small
\begin{tabular}{@{}llrrrr@{}}
\toprule
Training & Evaluation & R@2 & R@10 & R@50 & R@100 \\
\midrule
\multicolumn{6}{@{}l}{\emph{Qwen (single vector)}} \\
Pretrained & Original & 0.8 & 1.2 & 2.2 & 2.7 \\
LIMIT Train &  & 12.5 & 39.7 & 100.0 & 100.0 \\
Permuted &  & 6.7 & 25.8 & 99.8 & 100.0 \\
Fresh &  & 18.8 & 38.9 & 73.9 & 86.1 \\
Two &  & 28.4 & 57.1 & 85.6 & 86.6 \\
\addlinespace
Split & Split test & 4.8 & 18.0 & 100.0 & 100.0 \\
Two & Extended & 23.5 & 52.9 & 81.3 & 82.0 \\
\addlinespace
Pretrained & Three & 27.9 & 26.7 & 24.0 & 23.0 \\
Two &  & 84.6 & 82.9 & 89.9 & 90.6 \\
\addlinespace
LIMIT Train & Atomic & 0.0 & 0.0 & 0.0 & 0.0 \\
Two &  & 21.7 & 32.8 & 44.8 & 50.5 \\
\midrule
\multicolumn{6}{@{}l}{\emph{ColBERT (multi-vector)}} \\
LIMIT Train & Original & 100.0 & 100.0 & 100.0 & 100.0 \\
Two &  & 100.0 & 100.0 & 100.0 & 100.0 \\
\addlinespace
LIMIT Train & Atomic & 99.9 & 100.0 & 100.0 & 100.0 \\
Two &  & 100.0 & 100.0 & 100.0 & 100.0 \\
\bottomrule
\end{tabular}
\end{table}

Jina-v4 shows the same contrast between its two heads (\Cref{tab:jina-limit-finetuning}). Before finetuning, its MV head already reaches $\mathrm{Recall}@10$ of 77\% on Original LIMIT, compared with 1.5\% for its SV head. Joint finetuning on LIMIT Train raises the $\mathrm{Recall}@10$ of the SV head to 83\% and that of the MV head to 99\%; on Atomic LIMIT, the SV head reaches 48\%, whereas the MV head reaches 97\%. Finetuning on Two or Three LIMIT also improves both heads on Original LIMIT, but on Atomic LIMIT the SV head reaches a $\mathrm{Recall}@10$ of at most 18\%, whereas the MV head reaches 94--96\%. The comparison is not uniform: after Two LIMIT training, the SV head reaches a $\mathrm{Recall}@100$ of 100.0\% on Original LIMIT, compared with 99.9\% for MV.

\begin{table}[t]
\caption{Jina-v4 recall (\%) after joint finetuning of its SV and MV heads (exact full-corpus retrieval). Pretrained is the model before finetuning; LIMIT Train is the checkpoint after 40,000 updates, and Two and Three are the final checkpoints after 50,000 updates.}
\label{tab:jina-limit-finetuning}
\small
\setlength{\tabcolsep}{3.5pt}
\begin{tabular}{@{}llrrrrrr@{}}
\toprule
& & \multicolumn{3}{c}{SV} & \multicolumn{3}{c}{MV} \\
\cmidrule(lr){3-5}\cmidrule(l){6-8}
Training & Evaluation & R@2 & R@10 & R@100 & R@2 & R@10 & R@100 \\
\midrule
Pretrained & Original & 0.8 & 1.5 & 4.1 & 63.6 & 77.2 & 88.2 \\
 & Atomic & 2.7 & 4.0 & 8.9 & 76.2 & 85.4 & 91.9 \\
LIMIT Train & Original & 68.5 & 82.9 & 93.2 & 99.2 & 99.4 & 99.4 \\
 & Atomic & 39.2 & 47.7 & 58.5 & 97.0 & 97.4 & 97.5 \\
Two & Original & 72.0 & 96.0 & 100.0 & 99.7 & 99.9 & 99.9 \\
 & Atomic & 8.6 & 15.8 & 29.9 & 92.0 & 93.7 & 95.6 \\
Three & Original & 64.9 & 89.9 & 98.3 & 99.3 & 99.5 & 99.6 \\
 & Atomic & 10.5 & 18.4 & 35.0 & 93.9 & 96.0 & 97.0 \\
\bottomrule
\end{tabular}
\end{table}

\subsection{Catastrophic Forgetting}\label{sec:forgetting}

\paragraph{Finetuning on LIMIT\@.} Finetuning improves retrieval on LIMIT-like datasets (\Cref{sec:finetuning}), but it also lowers retrieval quality on MS MARCO, which the pretrained models handle well (\Cref{tab:limit-forgetting}). Qwen's $\mathrm{Recall}@10$ falls from 36.8\% to 15.7\% after finetuning on LIMIT Train, a relative decrease of 57.4\%, and its $\mathrm{Recall}@100$ falls by 39.7\%. ColBERT's $\mathrm{Recall}@10$ falls only from 65.1\% to 61.8\% after finetuning on LIMIT Train, a decrease of 5.1\%, and its $\mathrm{Recall}@100$ falls by 1.2\%. Both models improve on LIMIT, but the multi-vector model keeps nearly all of its MS MARCO recall, whereas the single-vector model loses more than half of its $\mathrm{Recall}@10$. Within Jina-v4, finetuning on LIMIT Train lowers the $\mathrm{Recall}@10$ of the SV head from 58.8\% to 21.8\%, a decrease of 62.9\%, but that of the MV head only from 49.0\% to 38.5\%, a decrease of 21.3\%.

\begin{table}[t]
\caption{MS MARCO recall (\%) after finetuning on LIMIT-like datasets. Change is relative to the pretrained model. The Jina-v4 SV and MV heads are trained jointly.}
\label{tab:limit-forgetting}
\small
\begin{tabular}{@{}llrrrr@{}}
\toprule
Model & Training & R@10 & Change & R@100 & Change \\
\midrule
Qwen & Pretrained & 36.8 &  & 63.1 &  \\
 & LIMIT Train & 15.7 & $-57.4$\% & 38.1 & $-39.7$\% \\
\midrule
ColBERT & Pretrained & 65.1 &  & 83.0 &  \\
 & LIMIT Train & 61.8 & $-5.1$\% & 82.0 & $-1.2$\% \\
\midrule
Jina-v4 SV & Pretrained & 58.8 &  & 86.2 &  \\
 & LIMIT Train & 21.8 & $-62.9$\% & 52.0 & $-39.6$\% \\
Jina-v4 MV & Pretrained & 49.0 &  & 78.9 &  \\
 & LIMIT Train & 38.5 & $-21.3$\% & 67.0 & $-15.1$\% \\
\bottomrule
\end{tabular}
\end{table}

\paragraph{Finetuning on HotpotQA\@.} The same pattern appears when the training data are real. We finetune Qwen, ColBERT and Jina-v4 on HotpotQA and evaluate them on MS MARCO (\Cref{tab:msmarco-forgetting}). Qwen's $\mathrm{Recall}@10$ falls from 42.7\% before finetuning to 13.5\% at epoch 8, a relative decrease of 68.3\%, whereas ColBERT's changes only from 67.1\% to 66.0\%, a decrease of 1.5\%. This experiment was evaluated separately from the LIMIT experiment, and its Qwen and ColBERT baselines differ from those in \Cref{tab:limit-forgetting}; we therefore compare each finetuned model only with the pretrained model of its own experiment.

\begin{table}[t]
\caption{MS MARCO Recall@10 (\%) during finetuning on HotpotQA. Epoch 0 is the pretrained model, and Change is the relative change at epoch 8. The Jina-v4 SV and MV heads are trained jointly.}
\label{tab:msmarco-forgetting}
\small
\begin{tabular}{@{}lrrrrrr@{}}
\toprule
& \multicolumn{5}{c}{Epoch} & \\
\cmidrule(lr){2-6}
Model & 0 & 2 & 4 & 6 & 8 & Change \\
\midrule
Qwen & 42.7 & 21.2 & 16.5 & 14.8 & 13.5 & $-68.3$\% \\
ColBERT & 67.1 & 67.0 & 66.6 & 66.3 & 66.0 & $-1.5$\% \\
\midrule
Jina-v4 SV & 58.8 & 53.7 & 52.5 & 51.9 & 51.6 & $-12.1$\% \\
Jina-v4 MV & 49.0 & 49.8 & 49.4 & 48.6 & 48.1 & $-1.8$\% \\
\bottomrule
\end{tabular}
\end{table}

\begin{figure}[t]
\centering
\includegraphics[width=\columnwidth]{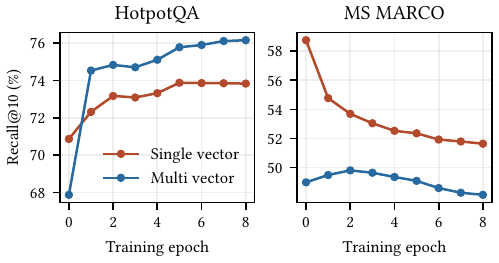}
\caption{Jina-v4 Recall@10 on HotpotQA and MS MARCO during joint finetuning on HotpotQA. Epoch 0 is the pretrained model; both heads come from the same model.}
\Description{Two line plots of Recall@10 against training epochs 0 to 8 for the single-vector and multi-vector heads of Jina-v4. On HotpotQA, both heads improve: the multi-vector head from about 68 to 76 percent and the single-vector head from about 71 to 74 percent. On MS MARCO, the single-vector head falls from about 59 to 52 percent, while the multi-vector head stays between about 48 and 50 percent.}
\label{fig:jina-forgetting}
\end{figure}
\paragraph{A controlled comparison.} By epoch 8, the HotpotQA $\mathrm{Recall}@10$ of Jina-v4 improves from 70.9\% to 73.8\% for SV and from 67.9\% to 76.2\% for MV. Meanwhile, MS MARCO $\mathrm{Recall}@10$ falls from 58.8\% to 51.6\% for SV and from 49.0\% to 48.1\% for MV, decreases of 7.12 and 0.86 percentage points (12.1\% and 1.8\% relative), respectively. The MV head starts below the SV head yet loses much less MS MARCO recall, so its stability is not explained by a higher starting point; SV nevertheless retains the higher absolute recall. \Cref{fig:jina-forgetting} shows both datasets.

\section{Limitations and Future Work}\label{sec:limitations}

The dimensionality bounds in \Cref{sec:signRankUpperBounds} are existence results for a fixed, known relevance matrix. However, we show evidence that current neural models are able to reach close to perfect recall on LIMIT after finetuning, even for single-vector models. It remains to be seen whether an encoder trained on data can find low-dimensional embeddings for LIMIT.

\section{Conclusion}\label{sec:conclusion}

In this work, we revisit the limitations of single-vector embeddings through the lens of LIMIT-like datasets. We argue that the poor performance of single-vector models on such datasets is unlikely to be explained solely by dimensionality considerations, since $2k+1$ dimensions suffice to represent a fixed relevance matrix with at most $k$ relevant documents per query. Instead, our analysis identifies several complementary factors, namely tokenization, linguistic similarity, domain shift and distance function misalignment, that collectively contribute to retrieval failures. Our experiments further demonstrate that although finetuning can substantially improve the performance of single-vector models, they remain significantly less robust than multi-vector embeddings. Across the settings we study, multi-vector models achieve stronger retrieval quality on LIMIT-like datasets and lose much less of their MS MARCO performance after finetuning, both across models (Qwen versus ColBERT) and between the jointly trained heads of Jina-v4. These results suggest that richer multi-vector representations, together with Chamfer-style similarity measures, are better suited for retrieval tasks with underlying combinatorial structure, although single-vector embeddings remain cheaper to store and index.

\bibliographystyle{ACM-Reference-Format}
\bibliography{references}

\appendix

\section{Proofs}\label{sec:proofs}

\subsection{ Details of Johnson-Lindenstrauss Construction}
\label{sec:proof-prop5}
We view each query/document as a subset of the universe of tokens $U$.
First observe that $d=|U|$ dimensions trivially suffice: assign to each token $u\in U$
the canonical unit vector $e_u\in\mathbb{R}^{|U|}$, and for any set $X\subseteq U$
define $\tilde{E}(X):=\sum_{u\in X} e_u$.
Then for any $Q,D$ we have $\langle \tilde{E}(Q),\tilde{E}(D)\rangle = |Q\cap D|$,
so for a fixed query $Q$ the ordering of documents by inner product agrees exactly
with the ordering by $|Q\cap D|$, equivalently by $\relset{Q}{D}=|Q\cap D|/|Q|$.

We now reduce the dimension using a Johnson--Lindenstrauss type projection.
Let $d$ be a parameter to be chosen later. By the Johnson--Lindenstrauss lemma
(equivalently, by standard concentration for random vectors), there exist unit vectors
$\hat{e}_1,\dots,\hat{e}_{|U|}\in\mathbb{R}^d$ such that for all $i\neq j$,
\begin{equation}\label{eq:almost-orth}
|\langle \hat{e}_i,\hat{e}_j\rangle| \;\le\; \varepsilon,
\qquad\text{where }\ \varepsilon = O\!\left(\sqrt{\frac{\log|U|}{d}}\right),
\end{equation}
and $\langle \hat{e}_i,\hat{e}_i\rangle = 1$.

For any $X\subseteq U$, define the (unnormalized) sum embedding $\tilde{E}(X)$ and the normalized embedding $E(X)$ by
\[
\tilde{E}(X):=\sum_{u\in X}\hat{e}_u,
\qquad
E(X):=\frac{\tilde{E}(X)}{\|\tilde{E}(X)\|_2}.
\]

Fix a query $Q$ with $|Q|\le \ell_Q$ and a document $D$ with $|D|\le \ell_D$.
Write $i:=|Q\cap D|$, $q:=|Q|$, and $m:=|D|$.
Expanding the dot product,
\[
\langle \tilde{E}(Q),\tilde{E}(D)\rangle
= \sum_{a\in Q}\sum_{b\in D}\langle \hat{e}_a,\hat{e}_b\rangle
= i \;+\!\!\sum_{\substack{a\in Q,\,b\in D\\ a\neq b}}\!\!\langle \hat{e}_a,\hat{e}_b\rangle,
\]
so by \eqref{eq:almost-orth},
\begin{equation}\label{eq:num}
\Big|\langle \tilde{E}(Q),\tilde{E}(D)\rangle - i\Big|
\;\le\; \varepsilon\,(qm-i) \;\le\; \varepsilon\,qm \;\le\; \varepsilon\,\ell_Q\ell_D.
\end{equation}
Similarly,
\[
\|\tilde{E}(Q)\|_2^2
= \sum_{a\in Q}\sum_{b\in Q}\langle \hat{e}_a,\hat{e}_b\rangle
= q \;+\!\!\sum_{\substack{a,b\in Q\\ a\neq b}}\!\!\langle \hat{e}_a,\hat{e}_b\rangle,
\]
hence $|\|\tilde{E}(Q)\|_2^2 - q|\le \varepsilon q(q-1)\le \varepsilon \ell_Q^2$,
and likewise $|\|\tilde{E}(D)\|_2^2 - m|\le \varepsilon \ell_D^2$.
Therefore (for $\varepsilon \ell_Q,\varepsilon \ell_D$ sufficiently small),
\begin{align*}
\|\tilde{E}(Q)\|_2 &= \sqrt{q}\,(1\pm O(\varepsilon \ell_Q)), \\
\|\tilde{E}(D)\|_2 &= \sqrt{m}\,(1\pm O(\varepsilon \ell_D)).
\end{align*}
Combining this with \eqref{eq:num} yields
\begin{align}\label{eq:dot-approx}
\langle E(Q),E(D)\rangle
&= \frac{\langle \tilde{E}(Q),\tilde{E}(D)\rangle}{\|\tilde{E}(Q)\|_2\,\|\tilde{E}(D)\|_2} \nonumber \\
&= \frac{i}{\sqrt{qm}}
 \pm O\!\left(\varepsilon\,\sqrt{qm}\right) \nonumber \\
&\subseteq \frac{i}{\sqrt{qm}}
 \pm O\!\left(\varepsilon\,\sqrt{\ell_Q\ell_D}\right).
\end{align}

Choose $d$ so that $\varepsilon \le \frac{1}{c\,\ell_Q\ell_D}$ for a sufficiently large absolute constant $c$.
Since $\varepsilon = O\!\big(\sqrt{\log|U|/d}\big)$, it suffices to take
\[
d = \Omega\!\big(\ell_Q^2\,\ell_D^2\,\log|U|\big).
\]
With this choice, the additive error term in \eqref{eq:dot-approx} is at most
$1/(4\sqrt{\ell_Q\ell_D})$. Hence, for any fixed query $Q$, the inner products
$\langle E(Q),E(D)\rangle$ order the documents in the same way as the main term
\[
\frac{|Q\cap D|}{\sqrt{|Q|\,|D|}}.
\]

This ordering agrees with $\relset{Q}{D}$ when all documents have the same length, so for
such datasets cosine similarity suffices. To match the set-theoretic relevance
$\relset{Q}{D}=|Q\cap D|/|Q|$ exactly for variable document lengths, we rescale documents
by their length: define $E'(D):=\sqrt{|D|}\,E(D)$ (equivalently, do not normalize the
document sum). The score $\langle E(Q),E'(D)\rangle$ is then an inner product rather than
a cosine similarity, and \eqref{eq:dot-approx} implies
\[
\langle E(Q),E'(D)\rangle
= \frac{|Q\cap D|}{\sqrt{|Q|}}
\ \pm\ O\!\left(\varepsilon\,|D|\,\sqrt{|Q|}\right),
\]
so for a fixed query $Q$, this score is strictly monotone in $|Q\cap D|$, and hence in
$\relset{Q}{D}$, once $d$ satisfies the same asymptotic lower bound above. Therefore, for each query $Q$ there exists a
threshold $\tau(Q)$ such that the induced ordering (or thresholding) agrees with the relevance rule
$\distset{Q}{D}\le \tau(Q)$, yielding the claimed sufficiency.

Taking the better of the two constructions ($d=|U|$ or the JL-based bound) gives
\[
d = O\!\left(\min\big\{|U|,\ \ell_Q^2\ell_D^2\log|U|\big\}\right),
\]
as required.

\section{Dataset Details}\label{sec:allExpts}

\subsection{LIMIT Dataset}
\label{sec:limit}

The \textsc{LIMIT} dataset introduced by \citet{weller2025} contains 50{,}000 synthetic person profiles, each listing 45 discrete ``like'' attributes. Its vocabulary of 1{,}848 attributes has two parts: 1{,}000 attributes used in the original queries and 848 additional filler attributes used to pad profiles. The original 1{,}000 queries have the form ``Who likes $x$?'' and each has exactly $k=2$ relevant documents. Throughout this appendix, relevance means membership of a complete attribute in a profile's attribute list, not a raw substring match in its text.

The original queries follow a \emph{dense qrel pattern}: their relevance sets are 1{,}000 distinct unordered pairs drawn from 46 core profiles, for which $\binom{46}{2}=1035$. The remaining 49{,}954 profiles contain only filler attributes. Profiles receive person names and padding attributes, but names do not determine relevance. The small LIMIT variant retains only the 46 core profiles; the full-corpus experiments described here use all 50{,}000 documents. We call the original 1{,}000-query evaluation set \emph{Original LIMIT}; whether its queries are held out depends on the training condition below.

\subsection{LIMIT Train}

\emph{LIMIT Train} uses the same 50{,}000-document corpus and one single-attribute query for each of the 848 filler attributes. These attributes are disjoint from the 1{,}000 Original LIMIT query attributes. Their relevance sets are dense: 2{,}484--2{,}792 documents per query (mean 2{,}650.94), rather than two. Thus LIMIT Train $\rightarrow$ Original LIMIT transfers from frequent filler attributes to the rare original query attributes in a shared corpus; it is not an 800/200 query split or a two-positive training task.

The 848 training queries and their relevance judgments are released with our code. Original LIMIT evaluation uses the original 1{,}000 queries and qrels of \cite{weller2025}.

\subsection{Split LIMIT}

\emph{Split LIMIT} instead partitions the original 1{,}000 queries into 800 training and 200 test queries, retaining the full original corpus and exactly two positives per query. The split is random but fixed, and is released with our code. Only the 200-query test set is held out from this training condition: evaluation on all of Original LIMIT also includes the 800 training queries.

\subsection{Atomic LIMIT}

\emph{Atomic LIMIT} consistently replaces the original attribute labels in both queries and profiles using a fixed one-to-one mapping to a replacement noun vocabulary. This mapping preserves attribute order and the two-positive relevance structure of the 1{,}000 queries over 50{,}000 documents; person names are replaced by ``This person.''

\subsection{Permuted LIMIT}

\emph{Permuted LIMIT} retains the same set of 1{,}000 query texts as Original LIMIT and applies a permutation $\pi$ to the original query attributes in the profiles, leaving filler attributes unchanged. Three within-profile orderings are retained for each original document. Consequently, the corpus contains 150{,}000 documents and six positives per query: three versions of each supporting profile. Its own query-ID/text mapping and induced qrels must be used together; equal query-text sets do not imply equal ID mappings. Training on Permuted LIMIT and evaluating on Original LIMIT changes the corpus/relevance mapping but does not hold out the query texts.

\subsection{Fresh LIMIT}

\emph{Fresh LIMIT} uses the same 1{,}000 original query texts and an independently sampled corpus of 50{,}000 profiles. Its construction assigns each queried attribute to one of 1{,}000 distinct pairs among 46 core profiles, preserving exactly two positives per query, and fills the other profiles using the filler vocabulary. This changes the supporting profiles and relevance assignment rather than merely permuting an existing assignment. All 1{,}000 queries are training queries, so transfer to Original LIMIT is not unseen-query-text generalization.

\subsection{Extended LIMIT}

\emph{Extended LIMIT} preserves the entire original 50{,}000-document corpus, including the 46 core profiles that support the original queries, and appends 50{,}000 profiles sampled only from the 848 filler attributes. The resulting corpus has 100{,}000 documents. Query texts, original document identities, and the two original positives per query remain unchanged; the added profiles contain none of the original query attributes and hence introduce no new positives. This is corpus enlargement around the same query/positive pairs, not a fresh sample of test queries or supporting profiles. The added profiles contain an extra separator character, which is retained in model inputs but is not treated as an attribute.

\subsection{Two LIMIT}
\emph{Two LIMIT} keeps the original 50{,}000 profiles and uses a set of 2{,}725 queries of the form ``Who likes $x$ or $y$?'' Each query contains two distinct attributes drawn from the full 1{,}848-attribute vocabulary, including the frequent filler attributes. This is a sampled set, not an enumeration of all $\binom{1848}{2}$ pairs. Relevance is the union $\operatorname{Rel}(x)\cup\operatorname{Rel}(y)$, not an AND condition. The qrels assign grade 1 to documents containing both attributes and 0.5 to those containing only one; both grades count as positive for recall. The queries have 3--5{,}393 positives each, so the two-positive property of Original LIMIT does not extend to this variant.

\subsection{Three LIMIT}
\emph{Three LIMIT} likewise keeps the original corpus and uses 1{,}805 queries of the form ``Who likes $x$, $y$, or $z$?'' with three distinct attributes from the full vocabulary. Relevance is $\operatorname{Rel}(x)\cup\operatorname{Rel}(y)\cup\operatorname{Rel}(z)$, with grade 1 for matching all three and 0.5 for matching a nonempty proper subset. There are 4--7{,}745 positives per query. The upper bound of six that would follow from querying only the 1{,}000 rare attributes does not apply to queries containing filler attributes.

Neither Two LIMIT nor Three LIMIT has a held-out within-variant split: their qrels define the training queries and any own-condition evaluation. Such evaluations are in-sample, as are own-condition evaluations for Permuted and Fresh LIMIT; they must be distinguished from transfer to another dataset and from Split LIMIT's held-out 200-query test. Within each dataset instance, the query set, attribute mapping, corpus and qrels are fixed.

\end{document}